\documentclass[12pt]{article}

\usepackage[a4paper, total={7in, 10in}]{geometry}

\RequirePackage{amsfonts, amssymb, amsthm, amsopn}
\RequirePackage{amsmath}
\RequirePackage{bm} 
\RequirePackage{pifont}
\RequirePackage{mathrsfs}
\RequirePackage[retainorgcmds]{IEEEtrantools}
\RequirePackage{fix-cm} 
\usepackage{xcolor}
\usepackage{xargs}
\usepackage[IEEEtran]{research18}
\usepackage{hyperref}
 \usepackage[caption=false,font=footnotesize]{subfig}

\newcommand{\defn}{\triangleq}

\let\Normal\relax
\newcommand{\Normal}{\mathcal{N}}

\newcommand{\M}{\mathcal{M}}
\newcommand{\T}{\mathcal{T}}

\newcommand{\C}{\mathcal{C}}

\newcommand{\Rh}{R_\textnormal{h}}

\newcommandx{\strong}[3][2=\epsilon,3=n]{\mathcal{A}^{* (#3)}_{#2} (#1)}
\newcommandx{\weak}[3][2=\epsilon,3=n]{\mathcal{A}^{ (#3)}_{#2} (#1)}
\newcommandx{\typen}[2][2=n]{\top ^{ (#2)}_{#1}}
\newcommandx{\contypen}[3][3=n]{\top ^{ (#3)}_{#1}(#2)}

\newcommandx{\alltypen}[2][2=n]{\set P_{#2}({#1})}
\newcommandx{\allcontypen}[2][2=n]{\set P_{#2}({#1})}
\newcommandx{\typeseqn}[2][2=n]{\top ^{ (#2)}({#1})}
\newcommandx{\allprob}[1]{\set P({#1})}

\let\vec\relax
\newcommand{\vec}{\mathbf}
\let\set\relax
\newcommand{\set}{\mathcal}

\newcommand{\bfX}{\mathbf{X}}

\newcommand{\bfx}{\mathbf{x}}
\newcommand{\bfy}{\mathbf{y}}

\newcommand{\bfyhat}{{\mathbf{\hat  y}}}
\newcommandx{\repbfyhat}[1][1=P]{\bfyhat_{#1}}
\newcommandx{\optrepbfyhat}[1][1=P]{\bfyhat^*_{#1}}

\newcommand{\dist}{\mathsf D}

\newcommand{\cardT}{|\mathcal T|}

\newcommandx{\admchannel}[1][1=\dist]{\set W^{\leq #1} }
\newcommandx{\admchanneln}[1][1=n]{\set W^{\leq \dist}_{#1} }
\newcommandx{\discontypen}[2][2=\dist]{\set W^{\leq #2}_n ({#1})}

\newcommandx{\admchannelpermn}[1][1=n]{\bar {\set W}^{\leq \dist}_{#1} }

\newcommandx{\GFp}[1][1=p]{\Field_{#1}}
\newcommandx{\Zp}[1][1=p]{\Integers_{#1}}

\newcommand{\snr}{\const{A}}

\begin{document}

\title{Message-Cognizant Assistance and Feedback for the Gaussian Channel}
\author{%
Amos Lapidoth, Ligong Wang, and Yiming Yan\thanks{The authors are with the Department of Information Technology and Electrical Engineering, ETH Zurich, 8092 Zurich, Switzerland (e-mail: \mbox{lapidoth@isi.ee.ethz.ch}; \mbox{ligwang@isi.ee.ethz.ch}; \mbox{yan@isi.ee.ethz.ch}). This work was supported by SNSF.}}
\date{}

\maketitle

\begin{abstract}
  A formula is derived for the capacity of the Gaussian channel with a
  benevolent message-cognizant rate-limited helper that provides a
  noncausal description of the noise to the encoder and decoder. This
  capacity is strictly larger than when the helper is message
  oblivious, with the difference being particularly pronounced at low
  signal-to-noise ratios. It is shown that in this setup, a feedback
  link from the receiver to the encoder does not increase
  capacity. However, in the presence of such a link, said capacity can
  be achieved even if the helper is oblivious to the transmitted message.
\end{abstract}

\section{Introduction}\label{sec:intro}


Complementing recent work on the capacity of the Gaussian channel with
a rate-limited message-oblivious helper \cite{bross2020decoder},
\cite{lapidoth2020encoder}, \cite{Entropy22}, \cite{merhav2021error},
we study here the message-cognizant helper. We focus on the case where
the help is to both encoder and decoder. The help is provided
noncausally to the communicating parties and comprises an $n \Rh$-long
message-dependent binary description of the noise sequence. For recent
results on DMCs with causal help see \cite{lapidoth2023wang}. Also highly
relevant to our work is \cite{khina2023modulation} which mainly focuses (but
not exclusively) on the transmission of a random parameter rather than
a message.

The channel we study is the classical discrete-time Gaussian noise
channel \cite{covertextbook} whose time-$k$ output $Y_{k}$
corresponding to the time-$k$ input $x_{k}$ is
\begin{equation}
  Y_{k} = x_{k} + Z_{k}
\end{equation}
where $\{Z_k\}$ are IID variance-$\sigma^{2}$ centered Gaussians, i.e.,
IID $\sim \Normal(0, \sigma^2)$. We assume that $\sigma > 0$, so 
noise is present.

A blocklength-$n$ rate-$R$ coding scheme with rate-$\Rh$
message-cognizant assistance comprises a message set
$\M= \{1, \dots, 2^{nR}\}$; a descriptions set
$\set{T}=\{1, \dots, 2^{n\Rh}\}$; a helper that produces the
assistance $T = h(Z^n, M)$ for some helping function
$h\colon \Reals^{n} \times \M \to \set{T}$ (where $A^{i}$ denotes
$(A_{1}, \ldots, A_{i})$ and $M$ is the transmitted message); an
encoder that produces the $n$-tuple
$X^n=\bfx(m, T) = (x_1(m,T), \dots, x_n(m,T))$ satisfying
\begin{equation}
 \E[T]{\sum_{k=1}^n x_k^2(m, T)} \leq n \PP  
\end{equation}
where $m \in \M$ is the message to be transmitted, 
$\E[T]{\cdot}$ denotes expectation with respect to $T$, and where
\begin{equation}
  \PP  > 0
\end{equation}
is some prespecified positive constant; and a decoder that produces
the message $\hat M=\psi(Y^n, T)$ for some decoding rule
$\psi\colon \Reals^{n} \times \set{T} \to \M$.

A rate $R$ is said to be achievable if there exists a sequence of
coding schemes as above (indexed by the blocklength) for which
\begin{equation}
  \lim_{n \to \infty} \Pr[M \neq \hat{M}] =  0
\end{equation}
when $M$ is drawn equiprobably from $\M$. The supremum of achievable
rates is the capacity $C$ we seek.

The feedback capacity is defined in an analogous way with the
transmitted $n$-tuple $\bfx(m,t)$ now having the form
\begin{equation*}
  \bfx(m,t,y^{n}) = \bigl( x_{1}(m,t), x_{2}(m,t,y_{1}), \ldots,
  x_{n}(m,t,y^{n-1}) \bigr).
\end{equation*}
It captures a scenario where, thanks to feedback link from the channel
output to the encoder, the time-$i$ transmitted symbol may depend not
only on the message $m$ and on the help $t$, but also on the
previously-received symbols $y^{i-1}$.

The message-oblivious helper capacity with feedback corresponds to a
scenario where there is a feedback link as above, but the helper is
message oblivious. The help now has the form $h(Z^{n})$ and the
time-$i$ channel input has the form $x_{i}(m,t,y^{i-1})$.

Our main result expresses the different capacities in
terms of $\Rh$ and the signal-to-noise ratio
\begin{equation}
  \snr \defn \frac{\PP}{\sigma^2}.
\end{equation}

\begin{theorem}[Message-Cognizant Helper]
  \label{thm:cognizant}
  On the Gaussian channel with a noncausal message-cognizant helper
  that assists both the encoder and the decoder
  \begin{align}
    C(\Rh)=\frac{1}{2}\log\left( 1+ \snr +2\sqrt{\snr (1-2^{-2\Rh})} \right) +\Rh. 
  \end{align}
  This remains the capacity also when a feedback link from
  the receiver to the encoder is added.
\end{theorem}
\begin{proof}
  The proof the direct part, which does not utilize the feedback link,
  can be found in Section~\ref{sec:achievability}. The converse, which
  is valid also in the presence of a feedback link, can be found in
  Section~\ref{sec:converse}.
  \end{proof}

  As Theorem~\ref{thm:cognizant} shows, feedback does not increase the
  capacity of the Gaussian message-cognizant helper capacity (when the
  help is provided to both encoder and decoder). It is, however, useful
  when the helper is message oblivious. In the absence of feedback,
  the message-oblivious helper capacity is \cite[Remark 5]{lapidoth2020encoder}
  \begin{equation*}
    \frac{1}{2}\log\left( 1+ \snr \right) +\Rh .
  \end{equation*}
  But, as the following theorem shows, feedback increases the capacity
  to that of message-cognizant helper:
\begin{theorem}[Message-Oblivious Helper with Feedback]
  \label{thm:oblivious}
  The capacity of the Gaussian channel with a feedback link from the channel output
  to the encoder and with a noncausal message-oblivious helper
  that assists both the encoder and the decoder is also
  \begin{equation*}
    \frac{1}{2}\log\left( 1+ \snr +2\sqrt{\snr (1-2^{-2\Rh})} \right) +\Rh
  \end{equation*}
  i.e., the same as that with a message-cognizant helper.
\end{theorem}
\begin{proof}
  In view of Theorem~\ref{thm:cognizant}, we only need a direct
  part. This is provided in Section~\ref{sec:fb_direct}, where we
  describe a feedback coding scheme with help that does not depend on
  the message.
  \end{proof}


    \section{Achievability}\label{sec:achievability}

    \subsection{In Broad Brushstrokes}
    
    We begin with a rough description of the coding scheme that
    ignores some of the technicalities. Let $f_{XYZ}$ be the centered
    multivariate Gaussian distribution under which $(X,Z)$ are of
    covariance matrix
    \begin{align}
      \begin{pmatrix}
        \PP & \sqrt{\PP}\sigma \rho\\
        \sqrt{\PP}\sigma \rho & \sigma^2
      \end{pmatrix}
    \end{align}
    and
    \begin{equation}
      Y = X + Z
    \end{equation}
    with probability one, where
    \begin{equation}
      \label{eq:def_rho}
      \rho= \sqrt{ 1- 2^{-2\Rh}}
    \end{equation}
    so 
    \begin{align}
      I(X;Z)&=h(Z)-h(Z|X)\\
            &=\frac{1}{2}\log (2\pi e \sigma^2) - \frac{1}{2}\log (2\pi e \sigma^2(1-\rho^2))\\
            &=\Rh. \label{eq:bin}
\end{align}



Generate $2^{n(R + \Rh)}$ codewords
$\{\bfx(m,t)\}_{(m,t) \in \M \times \set{T}}$ independently, with the
$n$ components of each being drawn IID $\Normal(0,\PP)$. If the message
to be transmitted is $M=m$, and if it observes the noise sequence
$Z^{n}$, the helper searches the $2^{n\Rh}$ codewords
$\{\bfx(m,t)\}_{t \in \set{T}}$ for a codeword $\bfx(m,t^{\star})$ that
is (weakly) jointly typical with $Z^{n}$ with respect to the
$XZ$-marginal $f_{XZ}$ of the above density $f_{XYZ}$. The helper is
very likely to find such $t^{\star}$ because, by our choice of $\rho$
\eqref{eq:def_rho}, $\Rh \approx I(X;Z)$. Having found $t^{\star}$,
the helper reveals it to the encoder and the decoder, with the former
now transmitting $\bfx(m,t^{\star})$. The decoder, for its part,
searches $\{\bfx(m',t^{\star})\}_{m' \in \M}$ for a some $\hat{M}$
for which $\bfx(\hat{M},t^{\star})$ is jointly typical with the
received sequence $Y^{n}$ with respect to the $XY$-marginal $f_{XY}$
of the above $f_{XYZ}$. Since the incorrect codewords are drawn
independently of $Y^{n}$, the decoding will succeed with high
probability when $R$ is approximately $I(X;Y)$ (where the latter is
computed with respect to $f_{XY}$). This mutual information is given by
\begin{align}
I(X;Y)&=h(Y)-h(Y|X)\\
&=h(Y)-h(Z|X)\\
&=\frac{1}{2}\log \Big(2\pi e \big(\PP+\sigma^2+2 \sqrt{\PP}\sigma \rho\big)\Big) - \frac{1}{2}\log (2\pi e \sigma^2(1-\rho^2)) \\
&=\frac{1}{2}\log \Big(1+\snr+2 \sqrt{\snr(1- 2^{-2\Rh})} \Big) +\Rh.
\end{align}


\subsection{A Geometric Approach}

For a more rigorous achievability proof, we propose a geometric
approach.

Let
$\partial \set{B}\bigl( \sqrt{n\PP} \bigr) =\{ \bfx \in
\Reals^{n}\colon \|\bfx\|^{2} = n \PP\}$ denote the
radius-$\sqrt{n\PP}$ $(n-1)$-dimensional Euclidean sphere in
$\Reals^{n}$, and let $\angle(\bfx,\bfy) \in [0,\pi]$ denote the angle
between the two (nonzero) vectors $\bfx,\bfy \in \Reals^{n}$ in the
sense that 
\begin{align}
\cos \angle(\bfx,\bfy) = \frac{\inner{\bfx}{\bfy}}{\norm{\bfx}\norm{\bfy}}. 
\end{align}

Fix $0 < \eps < \Rh$ (later to tend to zero), and let
$\theta_0\in[0,\pi/2]$ be such that
\begin{align}
\sin \theta_0 = 2^{-(\Rh-\epsilon)}.
\end{align}
Let $\C\subset \partial \set{B}\bigl( \sqrt{n\PP} \bigr)$ be a codebook of
$2^{n\Rh}$ codewords, indexed by $\set{T}$, with the covering property
that the caps of half-angle $\theta_{0}$ centered around the codewords
completely cover $\partial \set{B}\bigl( \sqrt{n\PP} \bigr)$. Such a
codebook exists whenever $n$ is large enough \cite{wyner97cap}, as we
henceforth assume.



Pick $|\M|$ random orthogonal transformations (rotations)
independently, each uniform according to the Haar measure, and index them by
the messages $m \in \M$. For each $m \in \M$, generate the set
$\C(m) = \{\bfx(m,t)\}_{t \in \set{T}}$ by applying the orthogonal
transformation corresponding to $m$ to each of the codewords in
$\C$.

Note that for each $m \in \M$, the set $\C(m)$---being the result
of rotating $\C$---also satisfies the covering property. This will be
important to keep in mind when we describe the transmission scheme.

Also note that, for each fixed $t \in \set{T}$, the codewords
$\{\bfX(m,t)\}_{m\in\M}$---which are the result of applying different
random rotations to the same element of $\C$---are independent and
uniformly distributed over the sphere. This observation will be
crucial to our analysis of the probability of error.

We next describe the transmission of some $m\in \M$. Upon observing the
noise $Z^n$, the helper seeks some $T \in \T$ such that the
angle between $\bfX(m,T)$ and $Z^n$ does not exceed
$\theta_0$. Such a $T$ exists because $\C(m)$ inherits the
covering property from $\C$. This $T$ (or one of those
satisfying the requirement) is revealed to both
the encoder and the decoder, with the former now transmitting
$\bfx(m,T)$.

The decoder---based on its observation $Y^{n}$ and the help
$T$---produces
\begin{equation}
  \hat{M} = \arg \min_{m' \in \M} \|Y^{n} - \bfx(m',T)\|.
\end{equation}


We next analyze the probability of error of our scheme. Set
\begin{align}
  \theta & = \angle(X^n, Z^n) \\
  & \leq \theta_0
\end{align}
where the inequality follows from our choice of $T$. In terms of
$\theta$,
\begin{align}
 \norm{Y^n}^2 &= \norm{X^n}^2 + \norm{Z^n}^2 + 2 \norm{X^n} \norm{Z^n}\cos\theta.
\end{align}

Setting
\begin{equation}
  \alpha = \angle(X^n, Y^n)
\end{equation}
we observe that
\begin{align}
\sin \alpha&=\frac{\norm{Z^n}}{\norm{Y^n}}\sin \theta\\
&=\Bigg(\sqrt{\frac{\norm{X^n}^2}{\norm{Z^n}^2} + 1 + 2\frac{\norm{X^n}}{\norm{Z^n}}\cos\theta}\Bigg)^{-1}\sin \theta. 
\end{align}
Recalling that $\norm{X^n}=\sqrt{n\PP}$, we obtain that, whenever
$\norm{Z^n}^2\leq n(\sigma^2+\epsilon)$,
\begin{align}
\sin \alpha&\leq \Bigg(\sqrt{ \frac{\PP}{\sigma^2+\epsilon}  + 1 + 2\sqrt{\frac{\PP}{\sigma^2+\epsilon}}\cos\theta}\Bigg)^{-1}\sin \theta\\
&\leq \Bigg(\sqrt{\frac{\PP}{\sigma^2+\epsilon}  + 1 + 2\sqrt{\frac{\PP}{\sigma^2+\epsilon} }\cos\theta_0}\Bigg)^{-1}\sin \theta_0\label{eq:theta}\\
&\defn \sin\alpha_0\label{eq:sin},
\end{align}
where \eqref{eq:theta}  holds because $\theta\leq \theta_0$ and
\eqref{eq:sin} defines $\alpha_0\in[0,\pi/2]$. 

Having verified that the condition 
$ \norm{Z^n}^2\leq n(\sigma^2+\epsilon)$ implies that 
$\angle(X^n,Y^n) \leq \alpha_0$ and that the condition $m'\neq m$
implies that $\bfX(m',T)$ is independent of $Y^{n}$ and uniformly
distributed  over the
sphere, we can bound the probability of error as follows:
\begin{align}
P_e(m)&\leq\bigPrv{\norm{Z^n}^2> n(\sigma^2+\epsilon)} +
        \bigPrv{\exists m'\neq m: \angle(\bfX(m',T),  Y^n) \leq \alpha_0}\\
&\leq\Prv{\norm{Z^n}^2> n(\sigma^2+\epsilon)} + 2^{nR} \cdot\frac{C_n(\alpha_0)}{C_n(\pi)} \label{eq:cn}\\
&= \Prv{\norm{Z^n}^2> n(\sigma^2+\epsilon)} + 2^{nR} \cdot 2^{-n(\log\sin\alpha_0+o(1))},\label{eq:pe}
\end{align}
where in \eqref{eq:cn}, we use $C_n(\phi)$ to denote the surface area of a spherical cap of 
half-angle $\phi$ on a unit $n$-sphere for $\phi\in[0, \pi]$; and \eqref{eq:pe} follows from \cite{wyner97cap}. The upper bound \eqref{eq:pe} decays to zero whenever 
\begin{align}
R&<-\log\sin\alpha_0\\
&= -\log\left(\Bigg(\sqrt{\frac{\PP}{\sigma^2+\epsilon}  + 1 + 2\sqrt{\frac{\PP}{\sigma^2+\epsilon}}\cos\theta_0}\Bigg)^{-1}\sin \theta_0\right)\\
&=\frac{1}{2}\log\bigg(\frac{\PP}{\sigma^2+\epsilon}  + 1 + 2\sqrt{\frac{\PP}{\sigma^2+\epsilon}}\cos\theta_0\bigg)-\log\sin \theta_0\\
&=\frac{1}{2}\log\bigg( \frac{\PP}{\sigma^2+\epsilon} + 1 + 2\sqrt{\frac{\PP}{\sigma^2+\epsilon}}\sqrt{1-2^{-2\Rh}}\bigg)+\Rh-\epsilon.
\end{align}
The direct part is now concluded by letting $\epsilon$ tend to zero by
employing the random-coding argument that guarantees that there exist
deterministic unitary transformation resulting in arbitrarily small
probability of error.

\section{Converse}\label{sec:converse}

We now prove a converse in the presence of a feedback link from the
channel output to the encoder.  Consider a message $M$ that is drawn
equiprobably from $\M$. Fano's inequality implies that, for any sequence of rate-$R$ coding schemes 
with rate-$\Rh$ message-cognizant assistance and vanishing
probabilities of error,  there exists some sequence $\{\delta_n\}$
tending to zero such that
\begin{IEEEeqnarray}{rCl}
nR-n\delta_n &=& H(M) - H(M|Y^n, T) \label{eq:fanodelta}\\
 &=& I(M; Y^n ,T)\\ 
 &=& I(M;Y^n|T)+I(M;T)\\
 &=& h(Y^n|T)-h(Y^n|M,T)+I(M;T)\\
 &=& h(Y^n|T)-h(Z^n)+   I(Z^n; T|M)+I(M;T)\label{eq:fanoh}\\
 &=& h(Y^n|T)-h(Z^n) +I(Z^n,M;T)\\
 &\leq& h(Y^n|T)-h(Z^n)+\log\cardT\\
 &\leq& h(Y^n)-h(Z^n)+\log\cardT\\
 &\leq& \sum_{k=1}^n h(Y_k)-h(Z^n)+\log\cardT\label{eq:fanoY},
\end{IEEEeqnarray}
where \eqref{eq:fanoh} can be justified as follows:
\begin{IEEEeqnarray}{rCl}
h(Y^n|M,T)  &=& \sum_{k=1}^n h(Y_k|M,T,Y^{k-1})\\
 &=& \sum_{k=1}^n h(Y_k-X_k|M,T,Y^{k-1}) \label{eq:fbxk}\\
 &=& \sum_{k=1}^n h(Z_k|M,T,Y^{k-1})\\
 &=& \sum_{k=1}^n h(Z_k) - \sum_{k=1}^n  I(Z_k; M,T,Y^{k-1})\\
 &=& \sum_{k=1}^n h(Z_k) - \sum_{k=1}^n  I(Z_k; M,T,Z^{k-1}) \label{eq:fbbij}\\
 &=& h(Z^n)- \sum_{k=1}^n  I(Z_k; M,T,Z^{k-1}) \label{eq:summisub}\\
 &=&h(Z^n)-   I(Z^n; T|M)
 \end{IEEEeqnarray}
 where \eqref{eq:fbxk} holds because $X_k$ is a function of
 $(M,T,Y^{k-1})$; \eqref{eq:fbbij} holds because there is a bijection
 between $(M,T,Y^{k-1})$ and $(M,T,Z^{k-1})$; and \eqref{eq:summisub}
 holds because
\begin{IEEEeqnarray}{rCl}
 \sum_{k=1}^n  I(Z_k; M,T,Z^{k-1}) &=&  \sum_{k=1}^n  I(Z_k; M,T| Z^{k-1}) \\
 &=& I(Z^n; M,T) \\
 &=&  I(Z^n; T|M). \label{eq:summi}
\end{IEEEeqnarray}

Having justified \eqref{eq:fanoY}, it remains to upper-bound its
RHS. We begin by bounding $I(X_{k};Z_{k})$ in two different ways. The
first upper-bounds it:
\begin{IEEEeqnarray}{rCl}
\sum_{k=1}^n I(X_k;Z_k) &\leq& \sum_{k=1}^n I(X_k, M, T, Z^{k-1};Z_k)  \\
&=&  \sum_{k=1}^n I(M, T, Z^{k-1};Z_k) \label{eq:fbxz}\\
&=& I(Z^n; T|M) \label{eq:summi2}\\
&\leq& \log \cardT\\
&=&n\Rh,  \label{eq:rhineq1}
\end{IEEEeqnarray}
where  \eqref{eq:fbxz} holds because $X_k$ is a function of $(M, T,
Z^{k-1})$; and \eqref{eq:summi2} follows from \eqref{eq:summi}.

The second lower-bounds it:
\begin{IEEEeqnarray}{rCl}
 I(X_k;Z_k) &=&  h(Z_k) - h(Z_k|X_k)\\
&=& \frac{1}{2}\log(2\pi e\sigma^2) - h(Z_k|X_k)\\
&\geq & \frac{1}{2}\log(2\pi e\sigma^2) -  \frac{1}{2}\log(2\pi e\sigma^2 (1-\rho_k^2))\label{eq:corr} \\
&= & -  \frac{1}{2}\log(1-\rho_k^2), \label{eq:rhineq2}
\end{IEEEeqnarray}
where in \eqref{eq:corr} we define $\rho_k\in[-1,1]$ to be the
correlation coefficient between $X_k$ and $Z_k$, and the inequality
follows from \cite[Problem 2.7]{gamalkimtextbook}.

From the two bounds \eqref{eq:rhineq1} and \eqref{eq:rhineq2}, 
\begin{IEEEeqnarray}{rCl}
\Rh&\geq &\frac{1}{n} \sum_{k=1}^n -  \frac{1}{2}\log(1-\rho_k^2)\\
&\geq &-  \frac{1}{2}\log\bigg(  1- \frac{1}{n}\sum_{k=1}^n \rho_k^2\bigg)\label{eq:fanolog}
\end{IEEEeqnarray}
where \eqref{eq:fanolog} follows from Jensen's inequality and the concavity of the logarithmic function. Hence, 
\begin{IEEEeqnarray}{rCl}
 \sum_{k=1}^n \rho_k^2
&\leq & n\big( 1- 2^{-2\Rh}\big).\label{eq:fanorh}
\end{IEEEeqnarray}

We now use this inequality to upper-bound the sum on the RHS of
\eqref{eq:fanoY}. For each $k$
\begin{IEEEeqnarray}{rCl}
\Var{Y_k} &=& \Var{X_k+Z_k}\\
&=& \Var{X_k}+ \Var{Z_k} + 2\sqrt{ \Var{X_k}\Var{Z_k}} \rho_k,
\end{IEEEeqnarray}
so
\begin{IEEEeqnarray}{rCl}
\sum_{k=1}^n h(Y_k) &\leq & \sum_{k=1}^n \frac{1}{2}\log(2\pi e\Var{Y_k}) \\
&\leq & n \cdot \frac{1}{2}\log\bigg(2\pi e\cdot \frac{1}{n} \sum_{k=1}^n \Var{Y_k}\bigg) \label{eq:fanoconc}\\
&= & n \cdot \frac{1}{2}\log\left(2\pi e \cdot\frac{1}{n} \sum_{k=1}^n \Big( \Var{X_k}+ \Var{Z_k} + 2\sqrt{ \Var{X_k}\Var{Z_k}} \rho_k \Big)\right)\\
&\leq & n \cdot \frac{1}{2}\log\left(2\pi e \cdot\frac{1}{n} \sum_{k=1}^n \Big( \E{X_k^2}+ \Var{Z_k} + 2\sqrt{  \E{X_k^2}\Var{Z_k}} \abs{\rho_k} \Big)\right)\\
&\leq & n \cdot \frac{1}{2}\log\left(2\pi e \bigg(\PP+\sigma^2 +\frac{2\sigma}{n} \sum_{k=1}^n \sqrt{ \E{X_k^2}} \abs{\rho_k} \bigg)\right)\\
&\leq & n \cdot \frac{1}{2}\log\left(2\pi e \Bigg(\PP+\sigma^2 +\frac{2\sigma}{n} \sqrt{ \sum_{k=1}^n  \E{X_k^2}} \sqrt{\sum_{k=1}^n \rho_k^2} \Bigg)\right)\label{eq:fanocauchy}\\
&\leq & n \cdot \frac{1}{2}\log\left(2\pi e \Big(\PP+\sigma^2 +\frac{2\sigma}{n} \sqrt{n\PP} \sqrt{ n\big( 1- 2^{-2\Rh}\big)} \Big)\right)\label{eq:fanorhsub}\\
&= & n \cdot \frac{1}{2}\log\left(2\pi e \Big(\PP+\sigma^2 +2\sigma \sqrt{\PP  \big( 1- 2^{-2\Rh}\big)} \Big)\right)
\end{IEEEeqnarray}
where \eqref{eq:fanoconc} follows from the concavity of  the logarithmic function; \eqref{eq:fanocauchy} follows from the Cauchy--Schwarz inequality; and \eqref{eq:fanorhsub} follows from \eqref{eq:fanorh}.

Continuing from \eqref{eq:fanoY}, 
\begin{IEEEeqnarray}{rCl}
nR-n\delta_n &\leq& \sum_{k=1}^n h(Y_k) -\frac{n}{2}\log(2\pi e \sigma^2) +n\Rh\\
&\leq& n \cdot \frac{1}{2}\log\left(2\pi e \Big(\PP+\sigma^2 +2\sigma \sqrt{\PP  \big( 1- 2^{-2\Rh}\big)} \Big)\right)-\frac{n}{2}\log(2\pi e \sigma^2) +n\Rh\\
&=& n \cdot \frac{1}{2}\log \Big(1+\snr +2 \sqrt{\snr\big( 1- 2^{-2\Rh}\big)} \Big)  +n\Rh.
\end{IEEEeqnarray}
Diving both sides of the inequality by $n$ and letting $n$ tend to
infinity yields the desired inequality
\begin{IEEEeqnarray}{rCl}
R &\leq&\frac{1}{2}\log \Big(1+\snr +2 \sqrt{\snr\big( 1- 2^{-2\Rh}\big)} \Big)  +\Rh. 
\end{IEEEeqnarray}

\section{A Feedback Scheme for a Message-Oblivious Helper}
\label{sec:fb_direct}

We next prove Theorem~\ref{thm:oblivious} by proposing a feedback
scheme with a helper that is incognizant of the message. The scheme is
reminiscent of the Schalkwijk-Kailath scheme \cite{sk1966,
  gallager2010sk}, albeit with only one use of the feedback link.  For
notational reasons, we consider transmission with blocks of length
$n+1$ and denote the corresponding channel inputs
$X_{0}, X_{1}, \ldots, X_{n}$ for positive integers $n$.
%
In broad brushstrokes, the idea is the following: We map the message
$m$ to $X_{0}$ using an injective mapping, so that from $X_{0}$ one
could recover the message. The problem is, of course, that the
receiver has no access to $X_{0}$ but only to $Y_{0}$, i.e., to the
sum of $X_{0}$ and the noise sample $Z_{0}$. This noise sample,
however, is known to the encoder as of time $1$, because it is simply
the difference between $Y_{0}$ and $X_{0}$, and the former is revealed
to the encoder after time zero. Moreover, the noise sample $Z_{0}$ is
known to the helper ahead of time, and \emph{a fortiori} as of time
$1$. We can then think about $Z_{0}$ as a new message that is to be
conveyed to the receiver using $X_{1}, \ldots, X_{n}$ with this
message being known to both helper and encoder. The idea is to then
invoke Theorem~\ref{thm:cognizant} for the transmission of $Z_{0}$ in
the time slots $1$ through $n$. The noise sample $Z_{0}$ is, of
course, continuous, so there are some quantization issues to be
addressed.


To address this, we use a construction somewhat similar to the one in
\cite[Section 5.2]{Entropy22}: For notational reasons we now assume
that the message set is $\{0, \dots, 2^{nR}-1\}$ (i.e., starting at
zero) and assume that $2^{nR}$ is an integer (in order to be able to
write $2^{nR}$ instead of $\lfloor 2^{nR}\rfloor$.

To convey some message $m \in \M$, the
encoder transmits at time zero the symbol
\begin{IEEEeqnarray}{rCl}
X_0 = X_{0}(m) = m \cdot \frac{\sqrt{\PP}}{2^{nR}}.
\end{IEEEeqnarray}
The decoder observes $Y_0=X_0+Z_0$, and the encoder---thanks to the
feedback link---can calculate $Z_0$ ($=Y_{0} - X_{0}$) after
time-$0$. 
Using the symbols $X_{1}, \ldots, X_{n}$, the encoder attempts to convey
to the decoder a quantized version of $Z_{0}$ or, more precisely, the integer
\begin{IEEEeqnarray}{rCl}
M' = \bigg\lfloor Z_0 \cdot \frac{2^{nR}}{\sqrt{\PP}}\bigg\rfloor  \mod 2^{nR}, 
\end{IEEEeqnarray}
which is a function of $Z_0$ and hence is also known to the helper. It
does so while ignoring the feedback link. The feedback link is thus
used by our scheme only to convey $Y_{0}$ to the encoder before
time-$1$.

Since, as of time-$1$, the integer $M' \in \M$ is known to both
encoder and helper, we can employ a blocklength-$n$ coding scheme with
message-cognizant helper as in Theorem~\ref{thm:cognizant} to send
$M'$ using $X_{1}, \ldots, X_{n}$ with arbitrarily small probability
of error. With the receiver's guess of $M'$ (based on
$Y_{1}, \ldots, Y_{n}$) denoted $\hat{M}'$, the receiver then guesses
that the transmitted message was
\begin{IEEEeqnarray}{rCl}
\hat M =  \bigg\lfloor Y_0 \cdot \frac{2^{nR}}{\sqrt{\PP}} - \hat{M}' \bigg\rfloor \mod 2^{nR}.
\end{IEEEeqnarray}

As we next argue, if the decoder recovers $M'$ correctly, i.e., if
$\hat{M}' = M'$, then its guess $\hat{M}$ of $m$ is correct, because
in this case
\begin{IEEEeqnarray}{rCl}
\hat M 
&=&  \bigg\lfloor \big(X_0 +Z_0\big)\cdot  \frac{2^{nR}}{\sqrt{\PP}} - M' \bigg\rfloor \mod 2^{nR}\\
&=& \bigg\lfloor m  +Z_0  \cdot   \frac{2^{nR}}{\sqrt{\PP}}  -  M' \bigg\rfloor \mod 2^{nR}\\
&=&  m  +\bigg\lfloor Z_0  \cdot   \frac{2^{nR}}{\sqrt{\PP}}  \bigg\rfloor  - M'\mod 2^{nR}\label{eq:modulo}\\
&=& m , 
\end{IEEEeqnarray}
where \eqref{eq:modulo} holds because  $m$ and $M'$ are both integers.  
The probability of error in reconstructing $m$ is thus upper bounded
by the probability of error in reconstructing  $M'$, which, by Theorem~\ref{thm:cognizant}, can be made arbitrarily small whenever
\begin{IEEEeqnarray}{rCl}
R< \frac{1}{2}\log \Big(1+\snr +2 \sqrt{\snr\big( 1- 2^{-2\Rh}\big)} \Big)  +\Rh. 
\end{IEEEeqnarray}
This proves the achievability part of Theorem~\ref{thm:oblivious}.

    \bibliographystyle{hieeetr}

    \bibliography{./bibliography_gauss.bib}

\end{document}